\documentclass[prx,aps,twocolumn,notitlepage,superscriptaddress,showpacs,nofootinbib]{revtex4-2}

\usepackage{enumerate,appendix}
\usepackage{amsmath, amsthm, amssymb,commath}
\usepackage{color,calc,graphicx}
\usepackage[usenames,dvipsnames,svgnames,table,cmyk,hyperref]{xcolor}
\usepackage[colorlinks]{hyperref}
\usepackage{optidef}
\hypersetup{
	colorlinks = true,
	urlcolor = {blue},
	citecolor = {blue},
	linkcolor= {blue}
}

\usepackage{graphicx}
\usepackage{amsmath}
\usepackage{latexsym}
\usepackage{bbm}

\usepackage[charter,cal=cmcal,sfscaled=false]{mathdesign}
\usepackage{booktabs}
\usepackage{multirow}
\usepackage{subcaption}
\usepackage{dcolumn}
\usepackage{mathrsfs}

\def\SU{\mathop{\rm SU}}
\newcommand{\frS}{\mathfrak{S}}

\def \be {\begin{equation}}
\def \ee {\end{equation}}

\newcommand{\Tr}{\mathrm{Tr}}

\def \interior{\mathrm{Int}\,}
\def\blambda{\bm{\lambda}}
\def \mix{\mathop{\rm mix}\nolimits}

\def \cH{{\cal H}}

\def \sofc2{{\cal S}({\mathbb C}^2)}

\def\>{\rangle}
\def\<{\langle}

\newtheorem{theorem}{Theorem}
\newtheorem{lemma}[theorem]{Lemma}

\newtheorem{corollary}[theorem]{Corollary}

\newtheorem{example}{Example}

\def\Label#1{\label{#1}\ [\ \text{#1}\ ]\ }
\def\Label{\label}

\begin{document}

\title{Another quantum version of Sanov theorem}

\author{Masahito Hayashi}\email{hmasahito@cuhk.edu.cn}
\affiliation{School of Data Science, The Chinese University of Hong Kong, Shenzhen, Longgang District, Shenzhen, 518172, China}
\affiliation{International Quantum Academy, Futian District, Shenzhen 518048, China}
\affiliation{Graduate School of Mathematics, Nagoya University, Nagoya, 464-8602, Japan}

\begin{abstract}
We study how to extend Sanov theorem to the quantum setting.
Although a quantum version of the Sanov theorem was proposed in 
Bjelakovic et al (Commun. Math. Phys., 260, p.659 (2005)),
the classical case of their statement is not the same as 
Sanov theorem
because Sanov theorem discusses the behavior of 
the empirical distribution when 
the empirical distribution is different from the true distribution,
but they studied a problem related to quantum hypothesis testing, whose classical version can be shown by 
classical Sanov theorem.
We propose 
another quantum version of Sanov theorem
by considering the quantum analog of 
the empirical distribution.
\end{abstract}

\maketitle

\section{Introduction}
Quantum Sanov theorem is known as a theorem for 
quantum hypothesis testing with a composite null hypothesis and a simple alternative hypothesis \cite{Bjelakovic,Notzel},
which is a natural extension of quantum Stein's lemma \cite{HP,ON}.
However, the original Sanov theorem has a
slightly different statement in the classical setting \cite{Bucklew,DZ,CT}.
That is, when $n$ data $X^n$ is subject to 
the independent and identical distribution of $Q$,
the original Sanov theorem
evaluates the probability that the empirical distribution of $X$ is $P$.
That is, this probability goes to zero exponentially, 
and its exponent is given as the relative entropy $D(P\|Q) $.
Using this fact, we can discuss the above type of hypothesis testing in the classical case.
Hence, the authors of the references \cite{Notzel,Bjelakovic}
considered that the above statement is a quantum extension of Sanov theorem.
That is, although the original Sanov theorem focuses on the empirical distribution,
the quantum Sanov theorem by \cite{Notzel,Bjelakovic}
does not consider any quantum extension of an empirical distribution.

Therefore, it is possible to consider another quantum version of 
Sanov theorem by considering a quantum extension of an empirical distribution.
In this paper, we define a quantum analog of an empirical distribution
by using a density matrix $\rho$
and discusses the case when
the true state is the $n$-fold tensor product of a density matrix $\sigma$.
However, when we simply measure each system over a given basis repetitively,
we cannot distinguish two different states that 
have the same diagonal elements for the basis ${\cal B}$.
In fact, the paper \cite{H-G} studied
Sanov theorem based on the empirical distribution
based on this method, 
In this paper, to avoid the above problem,
we employ the simultaneous measurement of 
the empirical distribution and the decomposition given by Schur duality
because the $n$-fold tensor product of a density matrix 
$\sigma$ satisfies the symmetry given by Schur duality.
That is, the use of the decomposition given by Schur duality is essential
to avoid the above problem.
Then, we evaluate the probability with respect to the outcome 
of the above simultaneous measurement. 

Indeed, many results in information theory for the classical system
employs empirical distribution.
For example, the strong converse for channel resolvability for fixed input
\cite{WH}
employs employs empirical distribution.
To exploit such approaches, we need to establish a quantum extension of 
Sanov theorem based on a quantum extension of 
empirical distribution.
In this way, this kind of result is strongly desired in the area of quantum information theory.

The remaining part of this paper is organized as follows.
First, Section \ref{S2} reviews the relation between 
the preceding results by \cite{Notzel,Bjelakovic}
and the original Sanov theorem.
Section \ref{S3} presents our main results.
Section \ref{S4} gives preparations for our analysis.
Section \ref{S5} proves Lemma \ref{LL1}.
Section \ref{S4D} gives two technical lemmas for 
information quantities.
Section \ref{S7} proves Lemma \ref{LL2}.
Section \ref{S8} proves the direct part of our main result.
Section \ref{S9} proves the converse part of our main result.
Section \ref{S10} gives the discussion.

\section{Review of existing results}\Label{S2}
We focus on a $d$-dimensional quantum system ${\cal H}$
spanned by the basis $\{|j\rangle\}_{j=1}^d$.
We consider the $n$-fold tensor product system ${\cal H}^{\otimes n}$.
To state the existing Sanov theorem,
we consider a set $S$ of density matrices on ${\cal H}$
and a density matrix $\sigma$ on ${\cal H}$.

We consider a $d$-dimensional Hilbert space ${\cal H}$ and a state $\rho$ on $\cH$.
We assume the $n$-fold independently and identical distributed (iid) condition
and consider the following two hypotheses;
\begin{align}
H_1:& \hbox{The state is }\sigma^{\otimes n},\\
H_0:& \hbox{The state is }\rho^{\otimes n} \hbox{ with }\rho \in S.
\end{align}
To address this problem, we focus on the following quantity.
\begin{align}
\beta_{\epsilon,n}(S\|\sigma):=
\min_{0\le T\le I}
\Big\{\Tr T \sigma^{\otimes n} \Big|
\max_{\rho\in S} \Tr (I-T)\rho^{\otimes n}\le \epsilon
\Big\}.
\end{align}
Then, we define the quantum relative entropy
\begin{align}
D(\rho\|\sigma):= \Tr \rho (\log \rho -\log \sigma).
\end{align}
The relation 
\begin{align}
\lim_{n \to \infty}-\frac{1}{n} \log \beta_{\epsilon,n}(S\|\sigma)
=\inf_{\rho \in S}D(\rho\|\sigma)\Label{BVDY}
\end{align}
is known as quantum Sanov theorem \cite{Bjelakovic,Notzel}.
The above theorem means that 
the choice of the test does not depend on a state $\rho \in S$ in the sense of quantum Stein's lemma.
In fact, 
when our measurement is required 
to achieve the optimal performance 
in the sense of quantum Stein's lemma,
the measurement does not depend on a state $\rho \in S$
\cite{H-02}.

Next, we review the original Sanov theorem in the classical case, i.e., 
when all possible states are diagonal state with respect to 
a basis ${\cal B}:=\{|v_j\rangle\}_{j=1}^d$.
We denote the set of diagonal states for 
the basis ${\cal B}$ by ${\cal S}[{\cal B}]$.
Also, we assume that our observation is based on the measurement based on 
the basis ${\cal B}$.
When our system is ${\cal H}^{\otimes n}$,
our observation is given as a basis 
$|v[x^n]\rangle:=|v_{x_1}, \ldots,v_{x_n}\rangle$
with a data $x^n=(x_1, \ldots,x_n)$.
In this case, we obtain the data $x^n$.
The state $\sum_{j=1}^n \frac{1}{n}|v_{x_j}\rangle \langle v_{x_j}|$
is called its empirical state under the basis ${\cal B}$, 
and is denoted by $\rho_{x^n,{\cal B}}$.
In the $n$-fold tensor product case, 
we denote the set of empirical states under the basis ${\cal B}$
by ${\cal S}_n[{\cal B}] $.
Then, for an empirical state $\rho\in {\cal S}_n[{\cal B}]$,
we consider the projection 
\begin{align}
T_{\rho,{\cal B}}^n:= \sum_{x^n:\rho_{x^n,{\cal B}}=\rho }|v[x^n]\rangle \langle v[x^n]|.
\end{align}
When $\rho \notin {\cal S}_n[{\cal B}]$,
$T_{\rho,{\cal B}}^n$ is defined to be $0$.
For any diagonal state $\sigma \in {\cal S}[{\cal B}]$, 
the eigenvalue of 
$\sigma^{\otimes n} T_{\rho,{\cal B}}^n$
is calculated by using the operator norm as
\begin{align}
\|\sigma^{\otimes n} T_{\rho,{\cal B}}^n\|=
\prod_{j=1}^d 
\langle v_{x_j}| \sigma| v_{x_j} \rangle^{n \langle v_{x_j}| \rho| v_{x_j} \rangle}
=e^{n \Tr \rho \log \sigma}.\Label{NM9}
\end{align}
Since $\Tr T_{\rho,{\cal B}}^n \cong e^{n H(\rho)}$,
we have
\begin{align}
\Tr \sigma^{\otimes n} T_{\rho,{\cal B}}^n \cong 
e^{n \Tr \rho \log \sigma}e^{n H(\rho)}
=e^{-n D(\rho\|\sigma)},\Label{BB1}
\end{align}
which is known as Sanov theorem \cite{Bucklew,DZ,CT}.

Given a subset $S\subset {\cal S}({\cal B})$,
we denote the interior of $S$ by $\interior(S)$, i.e., 
\begin{align}
\interior(S):=\{\rho\in S|
\exists \epsilon>0, 
U_\epsilon(\rho) \subset S\},
\end{align}
where $U_\epsilon(\rho)$ is the $\epsilon$-neighborhood of $\rho$
in the set ${\cal S}[{\cal B}]$.
Since 
a state $\rho \in \interior(S)$ satisfies
\begin{align}
\Tr \rho^{\otimes n} 
\sum_{\rho' \in S}T_{\rho',{\cal B}}^n \to 1 \hbox{ as }n \to \infty,
\end{align}
the converse part of Stein's lemma \cite{ON} implies  
\begin{align}
\lim_{n\to \infty}\frac{-1}{n}\log \Tr \sigma^{\otimes n} 
\sum_{\rho' \in S}T_{\rho',{\cal B}}^n \le D(\rho\|\sigma).\Label{BB2}
\end{align}
Since the number of possible empirical distributions is only a polynomial for $n$, 
the combination of \eqref{BB1} and \eqref{BB2} implies
\begin{align}
\inf_{\rho \in S} D(\rho\|\sigma)
\le
\lim_{n\to \infty}\frac{-1}{n}\log \Tr \sigma^{\otimes n} 
\sum_{\rho' \in S}T_{\rho',{\cal B}}^n 
\le \inf_{\rho \in \interior(S)} D(\rho\|\sigma),
\end{align}
which is also considered as Sanov theorem \cite{Bucklew,DZ}.
When $\inf_{\rho \in \interior(S)}D(\rho\|\sigma)=\inf_{\rho \in S}D(\rho\|\sigma)$, 
we have
\begin{align}
\Tr \sigma^{\otimes n} 
\sum_{\rho' \in S}T_{\rho',{\cal B}}^n \cong 
e^{-n \inf_{\rho \in S}D(\rho\|\sigma)}\Label{NMA}.
\end{align}
Since the converse part of \eqref{BVDY}, i.e., its part of $\le$, is guaranteed by
Stein's lemma, 
\eqref{NMA} implies \eqref{BVDY} when 
the states in $S$ and $\sigma$ are commutative with each other.
However, we cannot say that
\eqref{BVDY} implies \eqref{NMA}.

In addition, when 
we simply measure all subsystems of ${\cal H}^{\otimes n}$
with the basis ${\cal B}$,
it is impossible to distinguish the pure state $\sigma$ and 
the diagonal state $\sigma_{{\cal B}}$ for the basis ${\cal B}$ with the diagonal element of $\sigma$.
Therefore, it is necessary to study the quantum extension of 
\eqref{NMA}
separately the existing result \eqref{BVDY}, which 
can distinguish the pure state $\sigma$ and 
the diagonal state $\sigma_{{\cal B}}$.

\section{Main result}\Label{S3}
To formulate another quantum extension of Sanov theorem, 
we consider the case when 
the state $\sigma$ is not necessarily diagonal with respect to the basis ${\cal B}$.
For this aim, we prepare several notations.
Similar to the references \cite{H-01,H-02,KW01,CM,CHM,OW,HM02a,HM02b,Notzel,H-24,AISW,H-q-text},
we employ Schur duality of $\cH^{\otimes n}$. 
A sequence of monotone-increasing non-negative integers 
$\blambda=(\lambda_1,\ldots,\lambda_d)$ is called Young index.
Although many references \cite{H-q-text,Group1,GW} define the Young index as 
monotone-decreasing non-negative integers,
we define it in the opposite way for notational convenience. 
We denote the set of Young indices $\blambda$ with the condition
$\sum_{j=1}^d\lambda_j=n$ by $Y_d^n$.
Also, the set of probability distributions $p=(p_j)_{j=1}^d$ with the condition
$p_1\le p_2 \le \ldots \le p_d$ by ${\cal P}_d$.
For any density $\rho$ on ${\cal H}$, 
the eigenvalues of $\rho$ forms an element of ${\cal P}_d$, which
is denoted by $p(\rho)$.
Let ${\cal P}_d^n$ be the set of elements $p$ of ${\cal P}_d$ such that
$p_j$ is an integer time of $1/n$.
We define the majorization relation $\succ$ in two elements of 
${\cal P}_d$ and $Y_d^n$.
For two element $p,p'$ of ${\cal P}_d$, we say that
$p \succ p'$ when
\begin{align}
\sum_{j=1}^k p_j \ge \sum_{j=1}^k  p_j'
\end{align}
for $k=1, \ldots,d-1$.
We define the majorization relation $\succ$ for two elements of $Y_d^n$
in the same way.

As explained in \cite[Section 6.2]{H-q-text}, we have
\begin{align}
\cH^{\otimes n}=
\bigoplus_{\blambda \in Y_d^n}
{\cal U}_{\blambda} \otimes {\cal V}_{\blambda}.
\end{align}
Here,
${\cal U}_{\blambda}$ expresses the irreducible space of $\SU(d)$
and ${\cal V}_{\blambda}$ expresses the irreducible space of  
the representation $\pi$ of the permutation group $\frS_n$.
We define 
$d_{\blambda}:= \dim {\cal U}_{\blambda}$.
As shown in \cite[(6.16)]{H-q-text}, the dimension $d_{\blambda}$ is upper bounded as
\begin{align}
d_{\blambda}\le (n+1)^{\frac{d(d-1)}{2}}.\Label{NMI}
\end{align}

We denote the projection to the space ${\cal U}_{\blambda} \otimes {\cal V}_{\blambda}$ by 
$P_{\blambda}$.
For $\blambda \in Y_d^n$ and $\rho \in {\cal S}_n[{\cal B}] $,
we define the projection
\begin{align}
T_{\blambda,\rho,{\cal B}}^n 
:= P_{\blambda}
T_{\rho,{\cal B}}^n .
\end{align}
The relation 
$T_{\blambda,\rho,{\cal B}}^n \neq 0$ holds
if and only if 
\begin{align}
\blambda \succ n p(\rho).\Label{ZXN}
\end{align}
This is because the weight vectors in the range of 
$T_{\rho,{\cal B}}^n $ are restricted to 
vectors whose weight is majorized by the highest weight 
$\blambda$ as \eqref{ZXN} \cite{GW,H-q-text,Group1}.
We also define the set
\begin{align}
{\cal R}[{\cal B}] &:=\{
(p,\rho) \in {\cal P}_d \times {\cal S}[{\cal B}]|
p \succ p(\rho)\}, \\
{\cal R}_n[{\cal B}] &:={\cal R}[{\cal B}]\cap
({\cal P}_d^n \times {\cal S}_n[{\cal B}]).
\end{align} 
That is, the decomposition $\{T_{\blambda,\rho,{\cal B}}^n\}_{
(\frac{\blambda}{n},\rho)\in {\cal R}_n[{\cal B}] }$
is the joint measurement of 
the Schur sampling and 
the measurement to get the empirical distribution based on the basis ${\cal B}$
\cite{KW01,CM,CHM,OW,HM02a,HM02b,AISW}.
Notice that 
\begin{align}
|{\cal R}_n[{\cal B}]|\le (n+1)^{2(d-1)} \Label{SOA}
\end{align}
because $|{\cal P}_d^n|, |{\cal S}_n[{\cal B}]|\le (n+1)^{d-1} $.
The following subset of ${\cal R}[{\cal B}]$ plays a key role in our quantum analog of Sanov theorem.
For $\rho \in {\cal S}[{\cal B}]$ and $r>0$, we define
\begin{align}
S_{\rho,r} :=
\big\{  (p',\rho')\in {\cal R}[{\cal B}]\big|
r(p',\rho'\|\rho)\le r \big\} ,
\end{align}
where $r(p',\rho'\|\rho):=D(\rho'\| \rho)+H(\rho')-H (p') $.
The set $S_{\rho,r}$ is characterized as follows.
\begin{lemma}\Label{LL1}
For $\rho \in {\cal S}[{\cal B}]$ and $r>0$, we have
\begin{align}
\Tr \rho^{\otimes n} 
\Big(\sum_{(p',\rho') \in S_{\rho,r}^c\cap {\cal R}_n[{\cal B}]}T_{np',\rho',{\cal B}}^n \Big)
\le (n+1)^{\frac{(d+4)(d-1)}{2}}
e^{-nr},\Label{ER1}
\end{align}
and
\begin{align}
\lim_{n\to \infty}\frac{-1}{n}\log \Tr \rho^{\otimes n} 
\Big(\sum_{(p',\rho') \in S_{\rho,r}^c\cap {\cal R}_n[{\cal B}]}T_{np',\rho',{\cal B}}^n \Big)
=r.\Label{IBT3}
\end{align}
\end{lemma}
This lemma is shown in Section \ref{S5}.

To study the case with an arbitrary state $\sigma$, 
we recall the sandwich R\'{e}nyi relative entropy \cite{MDS,WWD}
\begin{align}
D_{1+s}(\rho\|\sigma)&:= \frac{\phi(-s|\rho\|\sigma )}{s}, \\
\phi(s|\rho\|\sigma )&:= \log \Tr (\sigma^{\frac{s}{2(1-s)}}\rho\sigma^{\frac{s}{2(1-s)}})^{1-s}.
\end{align} 
We have the following lemma.
\begin{lemma}\Label{LB1}
The map $s \mapsto \phi(s|\rho\|\sigma )$ is convex.
\end{lemma}
This lemma is shown in Section \ref{S4}.
The quantum relative entropy $D(\rho\|\sigma)$ is characterized as
\begin{align}
D(\rho\|\sigma)= \lim_{s\to 0} D_{1+s}(\rho\|\sigma).
\end{align}
Now, we consider another quantum extension of relative entropy as
\begin{align}
\hat{D}(\rho\|\sigma):= \lim_{s\to 0} \frac{-\phi(1-s|\sigma\|\rho)}{s}.
\end{align} 
Since the map $s \mapsto -\phi(1-s|\sigma\|\rho)$ is concave due to 
Lemma \ref{LB1},
the value $\frac{-\phi(1-s|\sigma\|\rho)}{s}$ is monotone decreasing
for $s>0$, which implies
\begin{align}
\hat{D}(\rho\|\sigma) \ge \frac{-\phi(1-s|\sigma\|\rho)}{s}.\Label{YIM}
\end{align}
In particular, substituting $1/2$ into $s$, we have
\begin{align}
\hat{D}(\rho\|\sigma) \ge 
-2\phi(1/2|\sigma\|\rho)= - 2 \log 
\Tr |\rho^{1/2}\sigma^{1/2} |.\Label{YIM2}
\end{align}
When $\rho\neq \sigma$, 
the fidelity $\Tr|\rho^{1/2}\sigma^{1/2} |$
is strictly smaller than $1$. Then,
$\hat{D}(\rho\|\sigma)$ is strictly positive due to \eqref{YIM}.

When $\rho$ and $\sigma$ are commutative with each other,
$\hat{D}(\sigma\|\rho)
={D}(\sigma\|\rho)$.
To describe the exponents, we define the function;
\begin{align}
\hat{B}_e(r|\rho\|\sigma):= 
\sup_{s\in (0,1)}
\frac{-(1-s) r- \phi(1-s|\sigma \| \rho)}{s}.
\end{align}
Its limit is characterized as follows.

\begin{lemma}\Label{LL2}
We have the relations
\begin{align}
D(\rho\|\sigma) &\ge \hat{D}(\rho\|\sigma) \Label{SMO2}, \\
\lim_{r \to +0}\hat{B}_e(r|\rho\|\sigma)
&= \hat{D}(\rho\|\sigma) \Label{SMO}.
\end{align}
\end{lemma}
This lemma is shown in Section \ref{S7}.
Then, we have the following versions of quantum Sanov theorem.
\begin{theorem}\Label{TH1}
We fix a basis ${\cal B}$.
For states $\rho_n \in {\cal S}_n[{\cal B}]$
and Young index $\blambda_n \in Y_d^n$
with the condition \eqref{ZXN},
we have
\begin{align}
&
\frac{-1}{n}\log \Tr \sigma^{\otimes n} 
T_{\blambda_n,\rho_n,{\cal B}}^n \notag\\
\ge &
\frac{-(1-s) r_n- \phi(1-s|\sigma \| \rho)}{s}
-\frac{d(d+3)}{2sn}\log (n+1) ,
\Label{ZPD2}
\end{align}
where $r_n:= r(\frac{\blambda_n}{n},\rho_n\| \rho)$.
Therefore, when $r_n\to r$,
taking the limit $n\to \infty$ and the supremum for $s$, 
we have
\begin{align}
\liminf_{n\to \infty}
\frac{-1}{n}\log \Tr \sigma^{\otimes n} 
T_{\blambda_n,\rho_n,{\cal B}}^n 
\ge 
\hat{B}_e (r
|\rho\|\sigma).\Label{ZPD}
\end{align}
In particular, when $r_n\to 0$, we have
\begin{align}
\liminf_{n\to \infty}
\frac{-1}{n}\log \Tr \sigma^{\otimes n} 
T_{\blambda_n,\rho_n,{\cal B}}^n 
\ge \hat{D}(\rho\|\sigma).\Label{ZPD6}
\end{align}
\end{theorem}

This theorem is shown in Section \ref{S8}.
The projection $T_{\blambda_n,\rho_n,{\cal B}}^n $
with $r_n\to 0$ takes a similar role of the empirical distribution 
in the original Sanov theorem \eqref{BB1}.
To get the same behavior as \eqref{BB1}, we need the opposite inequality.
For the opposite inequality, 
as explained in the following theorem,
we need to focus on the summand in the set 
$S_{\rho,r}\cap {\cal R}_n[{\cal B}]$.

\begin{theorem}\Label{TH2}
For $\rho \in {\cal S}[{\cal B}]$ and $r>0$, we have
\begin{align}
\lim_{n\to \infty}\frac{-1}{n}\log \Tr \sigma^{\otimes n} 
\Big(\sum_{(p',\rho') \in S_{\rho,r}\cap {\cal R}_n[{\cal B}]}T_{np',\rho',{\cal B}}^n \Big)
=\hat{B}_e(r|\rho\|\sigma).
\Label{ER2}
\end{align}
\end{theorem}
This theorem is shown in Section \ref{S9}.

Combining Lemma \ref{LL2} and Theorem \ref{TH1}, we have the following corollary.
\begin{corollary}\Label{Cor}
For any $0<R< \hat{D}(\rho\|\sigma)$, 
there exists a small $r>0$ such that
$\hat{B}_e(r|\rho\|\sigma)\ge R$, i.e., 
\begin{align}
\lim_{n\to \infty}\frac{-1}{n}\log \Tr \sigma^{\otimes n} 
\Big(\sum_{(p_n,\rho_n) \in S_{\rho,r}\cap {\cal R}_n[{\cal B}]}T_{np_n,\rho_n,{\cal B}}^n \Big)
\ge R.
\end{align}
\end{corollary}

When $\rho\neq \sigma$, 
$\hat{D}(\rho\|\sigma)$ is strictly positive.
Therefore, when 
we apply the measurement $\{T_{\blambda_n,\rho_n,{\cal B}}^n\}_{
(\frac{\blambda_n}{n},\rho_n)\in {\cal R}_n[{\cal B}] }$,
we can distinguish the state $\rho_n \in {\cal S}[{\cal B}]$ from any other states $\sigma$.
Hence, the information $(\frac{\blambda_n}{n},\rho_n)\in {\cal R}_n[{\cal B}] $ 
can be considered as a quantum version of 
an empirical distribution.
In fact, when we employ only the information $\rho_n \in {\cal R}_n[{\cal B}] $,
it contains the case when the Young index $\blambda_n \in Y_{d}^n$ has a small exponent
for the true state $\sigma$.
Therefore, it is important to use the Young index $\blambda_n$ as well.

\begin{example}
As a simple example, we consider the case when
$\sigma$ is a pure state
$|\psi_1\rangle\langle \psi_1| $.
\begin{align}
\phi(s|\rho\|\sigma )
:=& \log \Tr (
|\psi_1\rangle\langle \psi_1|
\rho
|\psi_1\rangle\langle \psi_1|
)^{1-s} \notag\\
=&
(1-s) \log \langle \psi_1|\rho |\psi_1\rangle
\end{align} 
Hence, the equality in \eqref{YIM2} holds, i.e.,
\begin{align}
\hat{D}(\sigma\|\rho)= 
-\log \langle \psi_1|\rho |\psi_1\rangle
=- 2 \log 
\Tr |\rho^{1/2}\sigma^{1/2} |
\end{align} 
while $D(\sigma\|\rho)= \infty$ unless $\sigma=\rho$.
This example shows that the equality in 
\eqref{SMO2}
does not hold in general.

Further, when $\rho$ is the diagonal state $\sigma_{\cal B}$,
$\hat{D}(\sigma\|\sigma_{\cal B})$ is strictly positive.
Hence, even though $\sigma_{\cal B}$ has the same diagonal element 
under the basis ${\cal B}$ as $\sigma$,
$\sigma$ and $\sigma_{\cal B}$ can be distinguished.
When the true state is a pure state $\sigma^{\otimes n}$,
the observed Young index $\lambda_n$ is always
$(n,0, \ldots, 0)$.
But, 
when the true state is the diagonal state $\sigma_{\cal B}^{\otimes n}$,
the probability that the observed Young index $\lambda_n$ is 
$(n,0, \ldots, 0)$ is quite small.
These two states can be distinguished by the observed Young index $\lambda_n$. 
\end{example}

\section{Preparation}\Label{S4}
We prepare several technical tools for our discussion.
Our proof employs 
the pinching map ${\cal E}_{{\cal B}}$, which is defined as
\begin{align}
{\cal E}_{{\cal B}}(X):= \sum_{\rho \in {\cal S}_n[{\cal B}]}T_{\rho,{\cal B}}^n X 
T_{\rho,{\cal B}}^n ,
\end{align}
where $X$ is a Hermitian matrices over ${\cal H}^{\otimes n}$.
Then, we introduce two TP-CP maps. For this aim,
we define ${\cal K}_n:= \bigoplus_{\blambda \in Y_d^n}
{\cal U}_{\blambda}$.
We define 
the TP-CP map $\Gamma_1$ from ${\cal S}(\cH^{\otimes n})$
to ${\cal S}({\cal K}_n)$ 
and
the TP-CP map $\Gamma_2$ from ${\cal S}({\cal K}_n)$ 
to ${\cal S}(\cH^{\otimes n})$
as
\begin{align}
\Gamma_1(\rho)&:= \sum_{\blambda \in Y_d^n}
\Tr_{{\cal V}_{\blambda}} P_{\blambda} \rho P_{\blambda} \\
\Gamma_2(\rho')&:= \sum_{\blambda \in Y_d^n}
\tilde{P}_{\blambda} \rho' \tilde{P}_{\blambda} \otimes \rho_{\blambda,\mix},
\end{align}
where $\tilde{P}_{\blambda}$ is the projection to 
${\cal U}_{\blambda}$ and 
$\rho_{\blambda,\mix}$ is the completely mixed state on 
${\cal U}_{\blambda}$.
When a state $\tilde{\rho}$ on the system ${\cal S}(\cH^{\otimes n})$
is permutation-invariant,
it has the following form;
\begin{align}
\tilde{\rho}&=
\bigoplus_{\blambda \in Y_d^n}
\tilde{\rho}_{\blambda} \otimes \rho_{\blambda,\mix}, \Label{BBR2}.
\end{align}
Then, we have
\begin{align}
\Gamma_2\circ \Gamma_1(\tilde{\rho})=\tilde{\rho}.
\end{align}
Since the states $\rho^{\otimes n}$ and 
${\cal E}_{{\cal B}}(\sigma^{\otimes n})$ are permutation-invariant, we have
\begin{align}
D_{1-s}(\rho^{\otimes n} \|  {\cal E}_{{\cal B}}(\sigma^{\otimes n}))
=
D_{1-s}(\Gamma_1(\rho^{\otimes n}) \| 
\Gamma_1({\cal E}_{{\cal B}}(\sigma^{\otimes n})).
\end{align}

The reference \cite[Next equation of (3.156)]{Springer} showed that
\begin{align}
& -\phi(1-s|{\cal E}_{{\cal B}}(\sigma^{\otimes n})\|
\rho^{\otimes n} )+(1-s)(d-1)\log (n+1) \notag\\
\ge &
-\phi(1-s|\sigma^{\otimes n}\|\rho^{\otimes n})
\ge
-\phi(1-s|{\cal E}_{{\cal B}}(\sigma^{\otimes n})\|\rho^{\otimes n} ).
\Label{XNF}
\end{align}
Since 
\begin{align}
\phi(1-s|\sigma^{\otimes n}\|\rho^{\otimes n})
=
n \phi(1-s|\sigma\|\rho),
\end{align}
we have
\begin{align}
\lim_{n\to \infty}\frac{1}{n}
\phi(1-s|{\cal E}_{{\cal B}}(\sigma^{\otimes n})\|\rho^{\otimes n} )
=\phi(1-s|\sigma\|\rho).
\Label{ZXL}
\end{align}
This relation shows Lemma \ref{LB1} as follows.

\begin{proof}[Proof of Lemma \ref{LB1}]
In the classical case, the function $s \mapsto \phi(1-s|\sigma\|\rho)$ is a convex function, which implies that
the function 
$s \mapsto\frac{1}{n}\phi(1-s|{\cal E}_{{\cal B}}(\sigma^{\otimes n})\|\rho^{\otimes n} )$ is a convex function for general states
$\sigma$ and $\rho$.
Hence, Eq. \eqref{ZXL} guarantees that the function 
$s \mapsto\phi(1-s|\sigma\|\rho)
)$ is a convex function for general states
$\sigma$ and $\rho$.
This shows Lemma \ref{LB1}.
\end{proof}

Since $T_{\blambda,\rho,{\cal B}}^n $ is permutation-invariant,
it is written as 
$\tilde{T}_{\blambda,\rho,{\cal B}} \otimes I_{\blambda} $,
where $I_{\blambda} $ is the identity operator on 
${\cal V}_{\blambda}$.
Also, $\Gamma_1({\cal E}_{{\cal B}}(\sigma^{\otimes n}))$ is commutative with
$\tilde{T}_{\blambda,\rho,{\cal B}}$.
We denotes the eigenvectors of 
$\tilde{T}_{\blambda,\rho,{\cal B}}
\Gamma_1({\cal E}_{{\cal B}}(\sigma^{\otimes n}))
\tilde{T}_{\blambda,\rho,{\cal B}}$ with non-zero eigenvalues
by $v_{\blambda,\rho,1}, \ldots, v_{\blambda,\rho,t(\blambda,\rho)}$.
Therefore, using \eqref{NM9}, we have
\begin{align}
&\langle v_{\blambda',\rho,j}| \Gamma_1(\rho^{\otimes n}) |v_{\blambda,\rho',j}\rangle 
=
\dim {\cal V}_{\blambda}
\|{T}_{\blambda,\rho',{\cal B}}
\rho^{\otimes n} {T}_{\blambda,\rho',{\cal B}}\|
\notag\\
=& \dim {\cal V}_{\blambda} e^{-n \Tr \rho' \log \rho}
=\dim {\cal V}_{\blambda} e^{-n D(\rho'\| \rho)-n H(\rho')}.
\end{align}

\begin{lemma}\Label{LL4}
For a Young index  $\blambda$, the dimension ${\cal V}_{\blambda}$
is evaluated as 
\begin{align}
e^{n H (\frac{\blambda}{n})} (n+1)^{-\frac{d(d+1)}{2}}
\le \dim {\cal V}_{\blambda}\le 
e^{n H (\frac{\blambda}{n})}.\Label{MNY}
\end{align}
\end{lemma}
The above lemma slightly improves the evaluation by \cite[Appendix B]{HM02b}.

\begin{proof}
For a Young index  $\blambda$, the dimension ${\cal V}_{\blambda}$
is calculated as 
\begin{align}
\dim {\cal V}_{\blambda}& =
\frac{n!}{\blambda!}
e(\blambda)\\
e(\blambda)&:=\prod_{j>i}\frac{\lambda_j-\lambda_i-i+j}{\lambda_j+j-i}<1,
\end{align}
where $ \blambda!:=\lambda_1!\lambda_2!\cdots \lambda_d!$.
For its detail, see \cite[Eq. (2.72)]{Group1}.
Using \cite[Theorem 2.5]{Springer}, we have 
\begin{align}
e^{n H (\frac{\blambda}{n})} (n+1)^{-(d-1)}
\le \frac{n!}{\blambda!}\le 
e^{n H (\frac{\blambda}{n})}.
\end{align}
The value $e(\blambda)$ is evaluated as
\begin{align}
&-\log e(\blambda)
=\sum_{j>i}\log \frac{\lambda_j+j-i}{\lambda_j-\lambda_i-i+j} 
\le  \frac{d(d-1)}{2}\log n.
\end{align}
Since $\frac{d(d-1)}{2}+(d-1)\le \frac{d(d-1)}{2}+d = \frac{d(d+1)}{2}$,
combining the above relations, we obtain \eqref{MNY}.
\end{proof}

\section{Proof of Lemma \ref{LL1}}\Label{S5}
We have an upper bound as follows.
\begin{align}
&\Tr \rho^{\otimes n} 
\Big(\sum_{(p',\rho') \in S_{\rho,r}^c\cap {\cal R}_n[{\cal B}]}T_{np',\rho',{\cal B}}^n \Big)\notag\\
=&
\sum_{ (\frac{\blambda_n}{n},\rho_n)\in S_{\rho,r}^c \cap {\cal R}_n[{\cal B}],j}
\langle v_{\blambda_n,\rho_n,j}| \Gamma_1(\rho^{\otimes n}) 
|v_{\blambda_n,\rho_n,j}\rangle \notag\\
\stackrel{(a)}{\le}&
\sum_{ (\frac{\blambda_n}{n},\rho_n)\in S_{\rho,r}^c\cap {\cal R}_n[{\cal B}] ,j}
e^{-n r(\frac{\blambda_n}{n},\rho_n\|\rho)
}  \notag\\
\stackrel{(b)}{\le}&
\sum_{ (\frac{\blambda_n}{n},\rho_n)\in S_{\rho,r}^c\cap {\cal R}_n[{\cal B}] }
d_{\blambda_n}
e^{-n r(\frac{\blambda_n}{n},\rho_n\|\rho)} \notag\\
\stackrel{(c)}{\le}&
(n+1)^{\frac{(d+4)(d-1)}{2}}e^{-nr}.\Label{IBT}
\end{align}
Here, $(a)$ follows from \eqref{MNY} in Lemma \ref{LL4}.
$(b)$ follows from the fact that the number of possible $j$ in this summand
is upper bounded by $d_{\blambda_n}$.
$(c)$ follows from \eqref{NMI} and \eqref{SOA}.
Hence, we obtain \eqref{ER1}.

We have a lower bound as follows.
\begin{align}
&\Tr \rho^{\otimes n} 
\Big(\sum_{(p',\rho') \in S_{\rho,r}^c\cap {\cal R}_n[{\cal B}]}T_{np',\rho',{\cal B}}^n \Big)\notag\\
=&
\sum_{ (\frac{\blambda_n}{n},\rho_n)\in S_{\rho,r}^c\cap{\cal R}_n[{\cal B}] ,j}
\langle v_{\blambda_n,\rho_n,j}| \Gamma_1(\rho^{\otimes n}) 
|v_{\blambda_n,\rho_n,j}\rangle \notag\\
\stackrel{(d)}{\ge}&
\sum_{ (\frac{\blambda_n}{n},\rho_n)\in S_{\rho,r}^c\cap{\cal R}_n[{\cal B}] ,j}
e^{-n r(\frac{\blambda_n}{n},\rho_n\|\rho)} 
 (n+1)^{-\frac{d(d+1)}{2}} \notag\\
\ge &
\sum_{ (\frac{\blambda_n}{n},\rho_n)\in S_{\rho,r}^c\cap{\cal R}_n[{\cal B}]}
e^{-n r(\frac{\blambda_n}{n},\rho_n\|\rho)} 
(n+1)^{-\frac{d(d+1)}{2}}.\Label{IBT2}
\end{align}
Here, $(d)$ follows from \eqref{MNY} in Lemma \ref{LL4}.
Since the number of elements of $S_{\rho,r}$ is polynomial for $n$,
Combining \eqref{NMI}, \eqref{IBT}, and \eqref{IBT2},
we obtain \eqref{IBT3}.

\section{Analysis for information quantities}\Label{S4D}
We prepare two useful lemmas related to 
our information quantities.
For $0<r < D(\sigma\|\rho)$, we define $s(r)$ as the solution of 
\begin{align}
s{\frac{d}{ds} \phi(1-s|\sigma \| \rho)}
=r+ \phi(1-s|\sigma \| \rho).\Label{ZMX}
\end{align}

\begin{lemma}\Label{L2}
We have
\begin{align}
\hat{B}_e(r|\rho\|\sigma)
= &\sup_{s\in (0,1)}
\frac{-(1-s) r- \phi(1-s|\sigma \| \rho)}{s}\notag\\
=&\frac{-(1-s(r)) r- \phi(1-s(r)|\sigma \| \rho)}{s(r)},
\Label{NMO}
\end{align}
and
\begin{align}
\frac{d}{dr} s(r)=&\frac{1}{s(r) \frac{d^2}{ds^2}\phi(1-s|\sigma \| \rho)|_{s=s(r)} }.\Label{PIV}
\end{align}
For $r \ge D(\sigma\|\rho)$, we have
\begin{align}
\hat{B}_e(r|\rho\|\sigma)
= &\sup_{s\in (0,1)}
\frac{-(1-s) r- \phi(1-s|\sigma \| \rho)}{s}=0.
\Label{NMO2}
\end{align}
\end{lemma}

\begin{proof}
Since the function $s \mapsto \phi(1-s|\sigma \| \rho)$
is $C^2$-continuous and convex, 
we have
\begin{align}
&\frac{d}{d s}
\frac{-(1-s) r- \phi(1-s|\sigma \| \rho)}{s}\notag\\
=&\frac{r+ \phi(1-s|\sigma \| \rho)
-s \frac{d}{ds} \phi(1-s|\sigma \| \rho)}{s^2}.\Label{NMS}
\end{align}
The derivative of the denominator is calculated as
\begin{align}
&\frac{d}{ds}
\Big(r+ \phi(1-s|\sigma \| \rho)
-s \frac{d}{ds} \phi(1-s|\sigma \| \rho)\Big)\notag\\
=&-s \frac{d^2}{ds^2} \phi(1-s|\sigma \| \rho) \ge 0.
\end{align}
Hence, the maximum of $\frac{-(1-s) r- \phi(1-s|\sigma \| \rho)}{s}$ is realized when
\eqref{NMS} equals zero.
That is, the maximum is realized when $s$ equals $s(r)$ that satisfies
\begin{align}
r+ \phi(1-s(r)|\sigma \| \rho)=
s(r) \frac{d}{ds} \phi(1-s|\sigma \| \rho)|_{s=s(r)}.\Label{NSX}
\end{align}
Taking the derivative for $r$ in \eqref{NSX}, we have
\eqref{PIV}.
\end{proof}

\begin{lemma}\Label{L4}
For $0<r < D(\sigma\|\rho)$, we have
\begin{align}
&\max_{r:0\le r< D(\sigma\|\rho)}
\Big(-(1-s(r_0))r\notag\\
&+ s(r_0) 
\frac{(1-s(r)) r+ \phi(1-s(r)|\sigma \| \rho)}{s(r)}\Big)\notag\\
=&\phi(1-s(r_0)|\sigma \| \rho).
\end{align}
\end{lemma}
\begin{proof}
We take the first derivative of the objective function as follows.
\begin{align}
&\frac{d}{dr}\Big(
-(1-s(r_0))r+ s(r_0) 
\frac{(1-s(r)) r+ \phi(1-s(r)|\sigma \| \rho)}{s(r)}\Big)\notag\\
=&
-(1-s(r_0))+ s(r_0)  \frac{(1-s(r))}{s(r)}\notag\\
&+s(r_0)
\frac{\partial}{\partial s}
\frac{(1-s) r+ \phi(1-s|\sigma \| \rho)}{s}\Big|_{s=s(r)} \frac{d s(r)}{dr}\notag\\
=&
-(1-s(r_0))+ s(r_0)  \frac{(1-s(r))}{s(r)}.
\end{align}

Then, we take the second derivative of the objective function as follows.
\begin{align}
&\frac{d^2}{dr^2}\Big(\!
-(1-s(r_0))r+ s(r_0) 
\frac{(1-s(r)) r+ \phi(1-s(r)|\sigma \| \rho)}{s(r)}\Big)\notag\\
=&
-\frac{s(r_0)}{s(r)^2}\frac{d s(r)}{dr}.
\end{align}
Hence, due to \eqref{PIV},
the above second derivative is negative.
The maximum is realized when the first derivative is zero.
When $r=r_0$, the first derivative is zero.
Hence, the desired statement is obtained.
\end{proof}

\section{Proof of Lemma \ref{LL2}}\Label{S7}
We employ Petz relative Renyi entropy 
$D_{1-s|P}(\rho\|\sigma):= \frac{1}{-s}\phi_P(s|\rho\|\sigma)$ \cite{Petz}, where
$\phi_P(s|\rho\|\sigma):= \log \Tr \rho^{1-s}\sigma^s$.
The information processing inequality for Petz relative Renyi entropy 
implies that
\begin{align}
n D_{1-s|P}(\rho\|\sigma)\ge 
\frac{-1}{s}\phi(1-s|{\cal E}_{{\cal B}}(\sigma^{\otimes n})\|\rho^{\otimes n} ).\Label{LAP}
\end{align}
Combining \eqref{ZXL} and \eqref{LAP}, we have
\begin{align}
D_{1-s|P}(\rho\|\sigma)\ge 
\frac{-1}{s}\phi(1-s|\sigma\|\rho).\Label{LAP2}
\end{align}
Taking the limit $s\to +0$, we obtain 
\eqref{SMO2} because the limit of Petz relative Renyi entropy 
is the quantum relative entropy $D(\rho\|\sigma)$.

To show \eqref{SMO}, we employ Lemma \ref{L2}.
Since $\phi(1-s|\sigma \| \rho)$ is convex,
$\frac{d^2}{ds^2}\phi(1-s|\sigma \| \rho)|_{s=s(r)}\ge 0$.
Thus, \eqref{PIV} guarantees that $\frac{d}{dr} s(r)\ge 0$.
$s(r)$ is monotone increasing for $0<r < D(\sigma\|\rho)$.
In particular, when $r=0$, $s(r)=0$ due to the relation \eqref{ZMX}.
Hence, when $r$ goes to zero, $s(r)$ goes to zero.
Using \eqref{ZMX}
and \eqref{NMO}, we have
\begin{align}
&\lim_{r \to 0}
\hat{B}_e(r|\rho\|\sigma)\notag\\
\stackrel{(a)}{=}&\lim_{r \to 0}
\frac{-(1-s(r)) r- \phi(1-s(r)|\sigma \| \rho)}{s(r)} \notag\\
\stackrel{(b)}{=}&\lim_{r \to 0}
\frac{-r- \phi(1-s(r)|\sigma \| \rho)}{s(r)} \notag\\
\stackrel{(c)}{=}&-\lim_{r \to 0}
\frac{d}{ds} \phi(1-s|\sigma \| \rho)|_{s=s(r)}\notag\\
=&-\lim_{s \to 0}
\frac{d}{ds} \phi(1-s|\sigma \| \rho)=\hat{D}(\rho\|\sigma).
\end{align}
Here, $(a)$ follows from \eqref{NMO} in Lemma \ref{L2}.
$(b)$ follows from the fact that $s(r)$ goes to zero.
$(c)$ follows from \eqref{ZMX}.
We obtain \eqref{SMO}.

\section{Direct part: Proof of Theorem \ref{TH1}}\Label{S8}
We have
\begin{align}
&e^{\phi(1-s|{\cal E}_{{\cal B}}(\sigma^{\otimes n})\|\rho^{\otimes n} )}
=e^{-s D_{1-s}(\Gamma_1(\rho^{\otimes n}) \| 
\Gamma_1({\cal E}_{{\cal B}}\sigma^{\otimes n}))}\notag\\
=&
\sum_{ (\frac{\blambda_n'}{n},\rho_n')\in {\cal R}_n[{\cal B}] ,j}
\langle v_{\blambda_n',\rho_n',j}| \Gamma_1(\rho^{\otimes n}) 
|v_{\blambda_n',\rho_n',j}\rangle^{1-s} \notag\\
&\cdot \langle v_{\blambda_n',\rho_n',j}| \Gamma_1({\cal E}_{{\cal B}}\sigma^{\otimes n}) 
|v_{\blambda_n',\rho_n',j}\rangle^s \notag\\
\ge &
\sum_j
\langle v_{\blambda_n,\rho_n,j}| \Gamma_1(\rho^{\otimes n}) 
|v_{\blambda_n,\rho_n,j}\rangle^{1-s}
\langle v_{\blambda_n,\rho_n,j}| \Gamma_1({\cal E}_{{\cal B}}\sigma^{\otimes n}) |v_{\blambda_n,\rho_n,j}\rangle^s \notag\\
\stackrel{(a)}{\ge}&
\sum_j
(e^{-n 
r(\frac{\blambda_n}{n},\rho_n\|\rho)} 
(n+1)^{-\frac{d(d+1)}{2}})^{1-s}\notag\\
&\cdot \langle v_{\blambda_n,\rho_n,j}| \Gamma_1({\cal E}_{{\cal B}}\sigma^{\otimes n}) |v_{\blambda_n,\rho_n,j}\rangle^s \notag\\
= &
(e^{-n r(\frac{\blambda_n}{n},\rho_n\|\rho)} 
(n+1)^{-\frac{d(d+1)}{2}})^{1-s}\notag\\
&\cdot 
\sum_j\langle v_{\blambda_n,\rho_n,j}| \Gamma_1({\cal E}_{{\cal B}}\sigma^{\otimes n}) |v_{\blambda_n,\rho_n,j}\rangle^s \notag\\
\stackrel{(b)}{\ge}&
(e^{-n 
r(\frac{\blambda_n}{n},\rho_n\|\rho)} 
(n+1)^{-\frac{d(d+1)}{2}})^{1-s}\notag\\
&\cdot 
\Big(\sum_j\langle v_{\blambda_n,\rho_n,j}| \Gamma_1({\cal E}_{{\cal B}}\sigma^{\otimes n}) |v_{\blambda_n,\rho_n,j}\rangle\Big)^s \notag\\
= &
e^{-n (1-s) r_n}
(n+1)^{-\frac{(1-s)d(d+1)}{2}}
\Big(\Tr \sigma^{\otimes n} 
T_{\blambda_n,\rho_n,{\cal B}}^n \Big)^s .\Label{NZP}
\end{align}
Here, $(a)$ follows from \eqref{MNY} in Lemma \ref{LL4}.
$(b)$ follows from the inequality $x^s+y^s \ge (x+y)^s$ for $x,y\ge 0$.
Thus, we have
\begin{align}
&\frac{-1}{n}\log \Tr \sigma^{\otimes n} 
T_{\blambda_n,\rho_n,{\cal B}}^n \notag\\
\stackrel{(c)}{\ge}& 
\frac{-(1-s) r_n- 
\frac{1}{n} \phi(1-s|{\cal E}_{{\cal B}}(\sigma^{\otimes n})\|\rho^{\otimes n} )}{s} \notag\\
&-\frac{(1-s)d(d+1)}{2sn}\log (n+1) \notag\\
\stackrel{(d)}{\ge}& 
\frac{-(1-s) r_n- \phi(1-s|\sigma \| \rho)}{s} \notag\\
&-\frac{(1-s)d(d+3)}{2sn}\log (n+1) \notag\\
\ge & 
\frac{-(1-s) r_n- \phi(1-s|\sigma \| \rho)}{s} 
-\frac{d(d+3)}{2sn}\log (n+1) ,
\Label{XNC}
\end{align}
where 
$(c)$ follows from \eqref{NZP}, and
$(d)$ follows from \eqref{XNF}.
and the inequality $-\frac{d(d+1)}{2}- (d-1)\ge -\frac{d(d+3)}{2}$.

To get \eqref{ZPD2}, taking the limit $n\to \infty$ with fixed 
$s\in (0,1)$, we have
\begin{align}
\liminf_{n\to \infty}
\frac{-1}{n}\log \Tr \sigma^{\otimes n} 
T_{\blambda_n,\rho_n,{\cal B}}^n 
\ge 
\frac{-(1-s) r- \phi(1-s|\sigma \| \rho)}{s}.
\end{align}
Taking the supremum for $s\in (0,1)$, we obtain \eqref{ZPD2}.


\section{Converse part: Proof of Theorem \ref{TH2}}\Label{S9}
We have
\begin{align}
&e^{\phi(1-s|{\cal E}_{{\cal B}}(\sigma^{\otimes n})\|\rho^{\otimes n} )}
=e^{-s D_{1-s}(\Gamma_1(\rho^{\otimes n}) \| 
\Gamma_1({\cal E}_{{\cal B}}\sigma^{\otimes n}))} \notag\\
=&
\sum_{ (\frac{\blambda_n}{n},\rho_n)\in {\cal R}_n[{\cal B}] ,j}
\langle v_{\blambda_n,\rho_n,j}| \Gamma_1(\rho^{\otimes n}) 
|v_{\blambda_n,\rho_n,j}\rangle^{1-s} \notag\\
&\cdot \langle v_{\blambda_n,\rho_n,j}| \Gamma_1({\cal E}_{{\cal B}}\sigma^{\otimes n}) 
|v_{\blambda_n,\rho_n,j}\rangle^s \notag\\
\stackrel{(a)}{\le}&
\sum_{ (\frac{\blambda_n}{n},\rho_n)\in {\cal R}_n[{\cal B}] ,j}
(e^{-n 
r(\frac{\blambda_n}{n},\rho_n\|\rho)
} 
)^{1-s}\notag\\
&\cdot \langle v_{\blambda_n,\rho_n,j}| \Gamma_1({\cal E}_{{\cal B}}\sigma^{\otimes n}) |v_{\blambda_n,\rho_n,j}\rangle^s \notag\\
= &
\sum_{ (\frac{\blambda_n}{n},\rho_n)\in {\cal R}_n[{\cal B}]}
(e^{-n 
r(\frac{\blambda_n}{n},\rho_n\|\rho)
} )^{1-s}\notag\\
&\cdot 
\sum_j\langle v_{\blambda_n,\rho_n,j}| \Gamma_1({\cal E}_{{\cal B}}\sigma^{\otimes n}) |v_{\blambda_n,\rho_n,j}\rangle^s \notag\\
\stackrel{(b)}{\le}&
\sum_{ (\frac{\blambda_n}{n},\rho_n)\in {\cal R}_n[{\cal B}]}
(e^{-n 
r(\frac{\blambda_n}{n},\rho_n\|\rho)
} 
)^{1-s}\notag\\
&\cdot d_{\blambda_n}
\Big(\sum_j\langle v_{\blambda_n,\rho_n,j}| \Gamma_1({\cal E}_{{\cal B}}\sigma^{\otimes n}) |v_{\blambda_n,\rho_n,j}\rangle\Big)^s \notag\\
= &
\sum_{ (\frac{\blambda_n}{n},\rho_n)\in {\cal R}_n[{\cal B}]}
(e^{-n 
r(\frac{\blambda_n}{n},\rho_n\|\rho)
} 
)^{1-s}
d_{\blambda_n}
\Big(\Tr \sigma^{\otimes n} 
T_{\blambda_n,\rho_n,{\cal B}}^n \Big)^s \notag \\
\le &
\sum_{ (\frac{\blambda_n}{n},\rho_n)\in {\cal R}_n[{\cal B}]}
(e^{-n r(\frac{\blambda_n}{n},\rho_n\|\rho)
} 
)^{1-s}
d_{\blambda_n} \\
&\cdot \Big(
\max_{ (\frac{\blambda_n'}{n},\rho_n')\in 
S_{\rho,r(\frac{\blambda_n}{n},\rho_n\|\rho)} \cap
{\cal R}_n[{\cal B}] }
\Tr \sigma^{\otimes n} 
T_{\blambda_n',\rho_n',{\cal B}}^n \Big)^s .\Label{ZCO}
\end{align}
Here, $(a)$ follows from \eqref{MNY} in Lemma \ref{LL4}.
$(b)$ follows from the fact that the number of possible $j$ in this summand
is upper bounded by $d_{\blambda}$.

We define
\begin{align}
R(r):
&=
\limsup_{n\to \infty}\frac{-1}{n}\log \Tr \sigma^{\otimes n} 
\Big(\sum_{(p',\rho') \in S_{\rho,r}\cap {\cal R}_n[{\cal B}]}T_{np',\rho',{\cal B}}^n \Big) \notag\\
&=
\limsup_{n\to \infty}\frac{-1}{n}\log 
\Big(\max_{(p',\rho') \in S_{\rho,r}\cap {\cal R}_n[{\cal B}]}\Tr \sigma^{\otimes n} T_{np',\rho',{\cal B}}^n \Big) .
\Label{ER5}
\end{align}
The second equality holds because 
\eqref{SOA}
guarantee that $\big|{\cal R}_n[{\cal B}]\big|$
increases polynomially for $n$.

Using \eqref{ZPD} and \eqref{NMO}, we have
\begin{align}
R(r) \ge
\frac{-(1-s(r)) r- \phi(1-s(r)|\sigma \| \rho)}{s(r)}.
\end{align}
Since 
\eqref{NMI} and \eqref{SOA}
guarantee that $d_{\blambda_n}$ and 
$\big|{\cal R}_n[{\cal B}]\big|$
increase polynomially for $n$,
using \eqref{ZXL}, \eqref{ZCO}, and \eqref{ER5}, we have
\begin{align}
&\phi(1-s|\sigma \| \rho)
\le
\sup_{r\ge 0}-(1-s)r- s R(r)\Label{BNB}
\end{align}
for $s \in (0,1)$.
Choosing $s=s(r_0)$, we have
\begin{align}
&\phi(1-s(r_0)|\sigma \| \rho)
\stackrel{(a)}{\le}
\sup_{r\ge 0}-(1-s(r_0))r- s(r_0) R(r)\notag\\
\stackrel{(b)}{\le} &
\sup_{r:0\le r< D(\sigma\|\rho)}\Big(-(1-s(r_0))r\notag\\
&+ s(r_0) 
\frac{(1-s(r)) r+ \phi(1-s(r)|\sigma \| \rho)}{s(r)}\Big)
\notag\\
\stackrel{(c)}{=} &\phi(1-s(r_0)|\sigma \| \rho).
\Label{LSU}
\end{align}
Here, $(a)$ follows from \eqref{BNB}.
$(b)$ follows from \eqref{ZPD} and \eqref{NMO2} in 
Lemma \ref{L2}.
$(c)$ follows from Lemma \ref{L4}, which is shown in Section \ref{S4D}.
Also, Lemma \ref{L4} guarantees that the supremum \eqref{LSU} is realized when $r=r_0$.

Therefore, we find that
\begin{align}
&\phi(1-s(r_0)|\sigma \| \rho)=
\sup_{r\ge 0}-(1-s(r_0))r- s(r_0) R(r)\notag\\
= &
-(1-s(r_0))r- s(r_0) R(r_0).\Label{CPR}
\end{align}
The combination of \eqref{NMO} and \eqref{CPR} implies 
\begin{align}
\hat{B}_e(r_0|\rho\|\sigma)=R(r_0).\Label{ER4}
\end{align}
Combining \eqref{ZPD} and \eqref{ER4}, we have
\begin{align}
&\limsup_{n\to \infty}\frac{-1}{n}\log \Tr \sigma^{\otimes n} 
\Big(\sum_{(p',\rho') \in S_{\rho,r}\cap {\cal R}_n[{\cal B}]}T_{np',\rho',{\cal B}}^n \Big)
\notag\\
=&
\liminf_{n\to \infty}\frac{-1}{n}\log \Tr \sigma^{\otimes n} 
\Big(\sum_{(p',\rho') \in S_{\rho,r}\cap {\cal R}_n[{\cal B}]}T_{np',\rho',{\cal B}}^n \Big). \Label{ER3}
\end{align}
The combination of \eqref{ER5}, \eqref{ER4}, and \eqref{ER3}
implies \eqref{ER2}.

Here, we should remark that 
the proof for \eqref{ZPD} works with $r=0$,
but the above proof does not work with $r_n \to 0$
as follows.
In this case
we need to choose $r_0=0$ in \eqref{LSU}.
That is, we need to choose $s=s(0)=0$ in \eqref{BNB}.
However, \eqref{ZCO} does not imply \eqref{BNB}
under $s=0$.
Therefore, our proof for Theorem \ref{TH2}
does not work with  $r_n \to 0$.

\section{Discussion}\Label{S10}
We have proposed another quantum version of Sanov theorem
by introducing a quantum extension of an empirical distribution.
In our analysis, a quantum empirical distribution is given as
the pair of an empirical distribution with respect to the given basis
and a Young index.
As given in Corollary \ref{Cor}, the exponent is given as $\hat{D}(\rho\|\sigma)$, which is smaller than the conventional quantum relative entropy
$D(\rho\|\sigma)$, which gives the optimal performance for simple quantum hypothesis testing as quantum Stein's lemma \cite{HP,ON}.
Several results in information theory for classical systems
employ empirical distributions, like \cite{WH}.
Hence, it can be expected that our result can be applied to 
their quantum extension.
Therefore, it is an interesting future study to derive 
various quantum extensions by using our result. 

\section*{Acknowledgement}
The author is supported in part by the National Natural Science Foundation of China (Grant No.
62171212).
The author is grateful for Professor Li Gao
to helpful comments.


\begin{thebibliography}{99}
\bibitem{Bjelakovic}
I. Bjelakovi\'{c}, J.-D. Deuschel, T. Kr\"{u}ger, 
R. Seiler, R. Siegmund-Schultze, and A. Szko{\l}a, 
``A Quantum Version of Sanov's Theorem,'' 
{\em Commun. Math. Phys.}, 
{\bf 260}, 659 -- 671 (2005). 

\bibitem{Notzel}
J. N\"{o}tzel,
``Hypothesis testing on invariant subspaces of the symmetric group: part I. Quantum Sanov's theorem and arbitrarily varying sources,''
{\em J. Phys. A: Math. Theor.}, 
{\bf 47}, 235303 (2014)

\bibitem{HP}	
F. Hiai and D. Petz, 
``The proper formula for relative entropy and its asymptotics in quantum
probability,'' 
{\em Comm. Math. Phys.}, 
{\bf 143}, 99 -– 114 (1991).

\bibitem{ON}
T. Ogawa and H. Nagaoka, 
``Strong converse and Stein’s lemma in quantum hypothesis testing,''
{\em IEEE Trans. Inf. Theory}, 
{\bf 46}, 2428 -– 2433 (2000).

\bibitem{Bucklew}
J. A. Bucklew,
{\em Large Deviation Techniques in Decision, Simulation, and Estimation}, 
John Wiley \& Sons, 1990.

\bibitem{DZ}
A. Dembo and O. Zeitouni,
{\em Large Deviations Techniques and Applications},
Stochastic Modelling and Applied Probability,
Springer (2010).

\bibitem{CT}
T. Cover and J. Thomas, 
{\em Elements of Information Theory (2 ed.)}. 
Hoboken, New Jersey: Wiley Interscience
(2006). 

\bibitem{H-G}
M. Horssen and M. Guta,
``Sanov and central limit theorems for output statistics of quantum Markov chains,''
{\em J. Math. Phys.} \textbf{56}, 022109 (2015).

\bibitem{WH}
S. Watanabe and M. Hayashi,
``Strong Converse and Second-Order Asymptotics of Channel Resolvability,''
{\em IEEE International Symposium on Information Theory (ISIT2014)}, Honolulu, HI, USA, June 29 - July 4, 2014. pp.1882

\bibitem{H-02}
M. Hayashi, ``Optimal sequence of quantum measurements in the sense of Stein's lemma in quantum hypothesis testing" 
{\em J. Phys. A: Math. Gen.}, 
{\bf 35}, 10759 -- 10773 (2002).

\bibitem{AISW}
J. Acharya, I. Issa, N. V. Shende, and A. B. Wagner, 
``Estimating Quantum Entropy,'' 
{\em IEEE Journal on Selected Areas in Information Theory}, 
{\bf 1}, 454 -- 468 (2020).

\bibitem{KW01}
M. Keyl and R.F.Werner, 
``Estimating the spectrum of a density operator,'' 
{\em Phys. Rev. A}, {\bf 64}, 052311 (2001).

\bibitem{CM}
M. Christandl and G. Mitchison, 
``The Spectra of Quantum States and the Kronecker Coefficients of the Symmetric Group,''
{\em Commun. Math. Phys.}, {\bf 261}, 789 –- 797 (2006).

\bibitem{CHM}
M. Christandl, A.W. Harrow, and G. Mitchison, 
``Nonzero Kronecker Coefficients and What They Tell us about Spectra,''
{\em Commun. Math. Phys.}, {\bf 270}, 575 -– 585 (2007). 

\bibitem{OW}
R. O'Donnell and J.  Wright, 
``Quantum Spectrum Testing,''
{\em Commun. Math. Phys.}, {\bf 387}, 1 -– 75 (2021). 

\bibitem{HM02a}
{M. Hayashi} and K. Matsumoto, ``Quantum universal variable-length source coding," 
{\em Phys. Rev. A}, {\bf 66}, 022311 (2002). 

\bibitem{HM02b}
{M. Hayashi}, ``Exponents of quantum fixed-length pure state source coding," 
{\em Phys. Rev. A}, {\bf 66}, 032321 (2002). 
{\em Phys. Rev. A}, {\bf 66}, 069901(E) (2002).

\bibitem{H-01}
M. Hayashi, ``Asymptotics of quantum relative entropy from a representation theoretical viewpoint," 
{\em J. Phys. A: Math. Gen.}, 
{\bf 34}, 3413 -- 3419 (2001).

\bibitem{H-24}
M. Hayashi and Y. Ito,
``Entanglement measures for detectability,''
arXiv: 2311.11189 (2024).

\bibitem{H-q-text}
M. Hayashi, 
{\em A Group Theoretic Approach to Quantum Information}, Springer (2017). 

\bibitem{MDS}
M. M\"{u}ller-Lennert, F.Dupuis, O. Szehr, S. Fehr,M. Tomamichel, 
``On quantum Renyi entropies: a new generalization and some properties,'' 
{\em J. Math. Phys.} {\bf 54}, 122203 (2013)

\bibitem{WWD}
M.M. Wilde, A. Winter, D. Yang, 
``Strong converse for the classical capacity of entanglementbreaking
and Hadamard channels via a sandwiched Renyi relative entropy,''
{\em Comm. Math. Phys.}
{\bf 331}(2), 593 (2014)

\bibitem{GW}
R. Goodman and N. R. Wallach, 
{\em Representations and Invariants of the Classical Groups}, 
Encyclopedia
of Mathematics and Its Applications, vol. 68 (Cambridge University Press, Cambridge,1999)

\bibitem{Group1}
M. Hayashi,
{\em Group Representation for Quantum Theory}, 
Springer (2017). 
(Originally published from Kyoritsu Shuppan in 2014 with Japanese.)

\bibitem{Springer}
M. Hayashi,
{\em Quantum Information Theory: Mathematical Foundation}, 
{\em Graduate Texts in Physics}, Springer (2017).
(First edition was published from Springer in 2006).

\bibitem{Petz}
D. Petz, 
``Quasi-entropies for finite quantum systems,'' 
{\em Rep. Math. Phys.} {\bf 23}, 57–65 (1986)

\end{thebibliography}
\end{document}